\documentclass{comjnl}

\usepackage{booktabs}
\usepackage[algo2e,ruled,vlined]{algorithm2e}

\SetAlFnt{\small}
\SetAlCapFnt{\small}
\SetAlCapNameFnt{\small}
\SetAlCapHSkip{0pt}
\IncMargin{-\parindent}

\usepackage{amsmath}
\allowdisplaybreaks
\usepackage{amssymb}
\usepackage{latexsym}
\usepackage{siunitx}

\usepackage{bm}
\usepackage{nicefrac}
\usepackage{booktabs}
\usepackage{array}
\usepackage{multirow}
\usepackage{threeparttable}
\usepackage{makecell}
\usepackage{stfloats}
\usepackage{graphicx}

\newenvironment{fminipage}%
{\begin{Sbox}\begin{minipage}}%
		{\end{minipage}\end{Sbox}\fbox{\TheSbox}}

\def\eps{\epsilon}

\def\abs#1{\left|#1  \right|}

\def\trace#1{\mathrm{Tr} \left(#1 \right)}
\def\norm#1{\| #1 \|}

\def\calG{\mathcal{G}}
\def\calL{\mathcal{L}}

\newcommand\bb{\boldsymbol{\mathit{b}}}

\newcommand\ee{\boldsymbol{\mathit{e}}}

\newcommand\qq{\boldsymbol{\mathit{q}}}

\def\tt{\boldsymbol{\mathit{t}}}
\newcommand\uu{\boldsymbol{\mathit{u}}}
\newcommand\vv{\boldsymbol{\mathit{v}}}

\newcommand\yy{\boldsymbol{\mathit{y}}}
\newcommand\zz{\boldsymbol{\mathit{z}}}
\newcommand\xx{\boldsymbol{\mathit{x}}}

\renewcommand\AA{\boldsymbol{\mathit{A}}}
\newcommand\BB{\boldsymbol{\mathit{B}}}

\newcommand\II{\boldsymbol{\mathit{I}}}
\newcommand\JJ{\boldsymbol{\mathit{J}}}

\newcommand\LL{\boldsymbol{\mathit{L}}}

\newcommand\QQ{\boldsymbol{\mathit{Q}}}

\renewcommand\SS{\boldsymbol{\mathit{S}}}

\newcommand\WW{\boldsymbol{\mathit{W}}}

\newcommand\ZZtil{\boldsymbol{\mathit{\tilde{Z}}}}

\newcommand\ZZ{\boldsymbol{\mathit{Z}}}

\newcommand\Otil{\widetilde{O}}

\newcommand{\one}{\mathbf{1}}

\newcommand{\kh}[1]{\left(#1\right)}

\DontPrintSemicolon
\SetKw{KwAnd}{and}
\SetFuncSty{textsc}
\SetKwInOut{Input}{Input\ \ \ \ }
\SetKwInOut{Output}{Output}

\usepackage{tabularx}
\usepackage{stfloats}

\usepackage{tabularx}

\begin{document}
\title{Diagonal of Pseudoinverse of  Graph Laplacian: Fast Estimation and Exact Results}
\author{Zenan~Lu} 
\author{Wanyue~Xu} 
\author{Zhongzhi~Zhang} \email{zhangzz@fudan.edu.cn}
\affiliation{Shanghai Key Laboratory of Intelligent Information
	Processing \& Shanghai Engineering Research Institute of Blockchain, School of Computer Science, Fudan University, Shanghai 200433, China}

\shortauthors{Z. Lu, W. Xu and Z. Zhang}

\keywords{Graph Laplacian,  pseudoinverse of graph Laplacian, graph and data mining, Kirchhoff index, Laplacian solver,node importance}

\begin{abstract}
The diagonal entries of pseudoinverse of the Laplacian matrix of a graph appear in many important practical applications, since they contain much information of the graph and many relevant quantities can be expressed in terms of them, such as Kirchhoff index and current flow centrality. However, a na\"{\i}ve approach for computing the diagonal of a matrix inverse has cubic computational complexity in terms of the matrix dimension, which is not acceptable for large graphs with millions of nodes. Thus,  rigorous solutions to the diagonal of the Laplacian  matrices for general graphs, even for particluar graphs are much less. In this paper, we propose a theoretically guaranteed estimation algorithm, which approximates all diagonal entries of the pseudoinverse  of a graph Laplacian in nearly linear time with respect to the number of edges in the graph. We execute extensive experiments on real-life networks, which indicate that our algorithm is both efficient and accurate. Also, we determine exact expressions for the diagonal elements of pseudoinverse of the Laplacian matrices for Koch networks and uniform recursive trees, and compare them with those obtained by our approximation algorithm. Finally, we use our algorithm to evaluate the Kirchhoff index of three deterministic model networks, for which the Kirchhoff index can be rigorously determined. These results further show the effectiveness and efficiency of our algorithm.
\end{abstract}

\maketitle

\section{Introduction}\label{sec:introduction}

As a typical representation of a graph, the Laplacian matrix  $\LL$ encapsulates much useful structural and dynamical information of the graph~\cite{Me94}. In addition to $\LL$ itself, its pseudoinverse $\LL^{\dagger}$ is also a powerful tool in network science~\cite{Ki18}, which arises in various aspects, such as random walks~\cite{SaMo11} and electrical networks~\cite{DoSiBu18}. Particularly, the diagonal entries of $\LL^{\dagger}$ appear frequently in diverse applications, for example, node centrality from both structural~\cite{EsHa10,RaZh11,RaZh13,VaDeCe17} and dynamical~\cite{SiBaBoMo16,SiBoBaMo18} perspectives. Moreover, a lot of other interesting quantities of a graph are also encoded in the diagonal entries of $\LL^{\dagger}$. For example, the sum of diagonal entries of $\LL^{\dagger}$ (trace of $\LL^{\dagger}$) is in fact the Kirchhoff index~\cite{KlRa93}, an invariant of a graph, which has found wide applications. First, it can serve as measures of the overall connectedness of a network~\cite{TiLe10},  the edge importance of complex networks~\cite{LiZh18}, and the robustness of first-order noisy networks~\cite{PaBa14} that has attracted much attension from the cybernetics community~\cite{ShCaHu18,SuYeQiCaCh19,SuLiZe20}. Besides, the popular current flow centrality~\cite{FiLe16,LiPeShYiZh19} or information centrality~\cite{PoYoScLe16} can also be represented in terms of the diagonal entries of $\LL^{\dagger}$.

In order to achieve better effects of the applications for the diagonal entries of $\LL^{\dagger}$ for a graph with $N$ nodes, the first step is to compute or evaluate the diagonal of $\LL^{\dagger}$. A straightforward computation of $\LL^{\dagger}$ involves inverting a perturbed Laplacian matrix $\LL+\JJ$, where $\JJ$ is the $N\times N$ matrix with every entry being $1$~\cite{GhBoSa08}, which costs $O(N^3)$ operations and $O(N^2)$ memories and thus is prohibitive for relatively large graphs with millions of nodes. In~\cite{RaZhBo14}, an incremental approach was proposed to compute $\LL^{\dagger}$. Although for general cases, this method performs better than the standard approach, for the worst case its computation cost is still $O(N^3)$. To reduce the computational complexity, some approximation algorithms were designed to estimate $\LL^{\dagger}$~\cite{BoFr12} or diagonal of a more general matrix~\cite{TaSa11,TaSa12,WuLaKaStGa16}, which have low complexity but no approximation guarantee. Thus, a theoretically guaranteed estimation algorithm for approximating $\LL^{\dagger}$ is imperative.

As shown above, many interesting quantities, such as Kirchhoff index, are closely related to the diagonal elements of matrix $\LL^{\dagger}$ of a graph $\calG$, the behaviors of which characterize various dynamical processes defined on $\calG$, including random walks, first-order noisy consensus~\cite{PaBa14,QiZhYiLi18,YiZhSt19}, and so on~\cite{Ne03}. However, the heavy demand on time and computer memory for inverting a matrix make it almost impossible to obtain exact results for the diagonal of $\LL^{\dagger}$ for a general large graph. It is thus of significant importance to seek for some particular graphs with some remarkable properties observed for real-world networks, the diagonal elements of $\LL^{\dagger}$ for which can be determined exactly by using computationally cheaper approaches. In addition to uncover the dependence of these primary quantities (e.g., Kirchhoff index) on the system size, rigorous results are also very useful as benchmark results for testing heuristic algorithms for computing the diagonal elements of a matrix inversion. Unfortunately, to the best of our knowledge, exact results for the diagonal of matrix $\LL^{\dagger}$ for scale-free small-world graphs is lacking, in spite of the fact that scale-free and small-world properties are ubiquitous in real networked systems~\cite{Ne03}.

In this paper, we focus on approximation method for computing the diagonal of the pseudoinverse $\LL^{\dagger}$ of a general graph, as well as exact results for diagonal of $\LL^{\dagger}$ for two special graphs. Our main contributions are as follows.

\begin{itemize}
	\item We introduce an approximation algorithm to compute the diagonal of the pseudoinverse $\LL^{\dagger}$ of the Laplacian matrix $\LL$ of a graph with $N$ nodes and $M$ edges. 
	It returns an approximation of the actual diagonal entries of $\LL^{\dagger}$ in nearly linear time with respect to $M$.  Moreover, we prove that our approach has a guarantee for  error bounds with a high probability.
	
	\item We perform extensive experiments on real-world networks, the results of which show that compared to the direct approach for inverting matrix, the proposed approximation method is both effective and efficient.
	
	\item We derive exact expressions for the diagonal elements of $\LL^{\dagger}$ for two deterministic networks, the scale-free small-world Koch networks~\cite{ZhZhXiChLiGu09} and the uniform recursive trees~\cite{LiZh13}, and compare them with the results obtained by our approximation algorithm.
	
	\item We apply our algorithm to estimate the Kirchhoff index of three networks, Koch networks, uniform recursive trees, and pseudofractal scale-free web~\cite{DoGoMe02,ShLiZh17}, the Kirchhoff index of which can be explicitly determined. These further demonstrate the efficiency and accuracy of our algorithm.
\end{itemize}

\section{Preliminaries}\label{}
In this section, we briefly introduce some basic concepts about graphs, including spanning  rooted forest, electrical network, resistance distance, Laplacian matrix and its pseudoinverse, spectral properties of Laplacian and its pseudoinverse, interpretations  of pseudoinverse for graph Lapcacian.

\subsection{Graphs and Spanning  Rooted Forests}

Let $\calG=(V,E,w)$ denote a connected undirected weighted graph or network,  where $V$ is the set of nodes (vertices)  $E$ is the set of edges (links), and  $w: E\to \mathbb{R}_{+}$ is the positive edge weight function, with $w_e$ being the weight for edge $e$. Then, there are total $N=|V|$ vertices and $M=|E|$ edges in graph $\calG$. We use $u \sim v$ to indicate that two vertices $u$ and $v$ are connected by an edge. Let $w_{\max}$ and $w_{\min}$ denote the maximum edge weight and minimum edge weight, respectively. Namely, $w_{\max}=\max_{e\in E} w_e $ and $w_{\min}=\min_{e\in E} w_e$.

For a graph  $\calG=(V,E,w)$,  a subgraph  $\mathcal{H}$ of  $\calG$  is a graph,  the node and edge sets of which are subsets of $V$ and  $E$, respectively. The product of the weights of the edges in  $\mathcal{H}$ is called the weight of  $\mathcal{H}$, denoted by  $ \varepsilon(\mathcal{H})$. The weight of a subgraph $\mathcal{H}$  with no edges is set to be 1. For any nonempty set $S$ of subgraphs, its weight is  defined as  $\varepsilon(S) =\sum_{\mathcal{H} \in S}  \varepsilon(\mathcal{H})$. The weight of the empty set is set to be zero~\cite{ChSh98}. A spanning subgraph $\calG'$  of  $\calG$ is a subgraph of  $\calG$ with the same node set $V$ and an edge set $E' \subseteq  E$.  A spanning tree  of $\calG$  is a  spanning subgraph of  $\calG$ that is a tree. A spanning forest on $\calG$  is a spanning subgraph  of $\calG$ that is a disjoint union of trees.  Here an isolated vertex is considered as a tree. A spanning rooted forest of  $\calG$ is a  spanning forest of  $\calG$ with a particular node (a root) marked in each tree. Let $\mathcal{F}_x$, $x \geq 1$, be the set of spanning rooted forests of $\calG=(V,E,w)$,  with each spanning rooted forest having  exactly $x$ trees.  Let $\mathcal{F}_x^{ij}$ be the set of spanning rooted forests of  $\calG=(V,E,w)$ with $x$ trees,  one of which is rooted at  $i$ and contains $j$.  And let $\mathcal{F}_x^{ii}$ be the set of spanning rooted forests  of  $\calG=(V,E,w)$ containing $x$ trees,  with one tree including $i$ and being rooted at  $i$.

\subsection{Electrical Networks and Resistance Distances}

For an arbitrary  undirected weighted graph  $\calG=(V,E,w)$, we can define an electrical network $\bar{\calG}=(V,E,r)$, which is obtained from $\calG$  by looking upon edges as resistors and considering vertices as junctions between resistors~\cite{DoSn84}, with the resistor of an associated  edge $e$ being $r_e=w_e^{-1}$.  For graph  $\calG$, the resistance distance $r_{ij}$ between two vertices $i$ and $j$  is defined as the effective resistance between $i$ and $j$ in the corresponding  electrical network $\bar{\calG}$~\cite{KlRa93}, which is equal to the potential difference between $i$ and $j$ when a unit current enters one vertex and leaves the other one. The  resistance distances of a graph have many interesting properties. For example, they obey the following Foster's theorem \cite{Te91}.
\begin{lemma}\label{Foster}
Let $\calG=(V,E,w)$ be a simple connected graph with $N$ nodes. Then the sum of  weight times resistance distance over all pairs of adjacent vertices in  $\calG$  satisfies
    \begin{equation*}
        \sum_{i \sim j}w_{ij} \times r_{ij}=N-1.
    \end{equation*}
\end{lemma}

Two key quantities related to  resistance distances  are the  resistance distance of a given node and Kirchhoff index.
\begin{definition}
For an  undirected weighted graph  $\calG=(V,E,w)$,  the  resistance distance  $r_i$ of a  node $i$ is the sum of resistance distances between $i$ and all other nodes:
\begin{equation*}
r_i:=\sum^{N}_{j=1}r_{ij}\,.
\end{equation*}
The Kirchhoff index $\delta$ of $\calG$ is the sum of resistance distances over all the $N(N-1)/2$ pairs of nodes:
\begin{equation*}
\delta:=\sum_{i=1}^{N}\sum_{j=i}^{N}r_{ij}\,.
\end{equation*}
\end{definition}
Both quantities appear in practical applications. The reciprocal of $r_i$ times $N$ is in fact the current flow centrality~\cite{BrFl05,LiPeShYiZh19} of node $i$, which is equivalent to information centrality~\cite{StZe89,ShYiZh18} of $i$.  While the Kirchhoff index can be applied to measure  the overall connectedness of a network~\cite{TiLe10}, the robustness of first-order noisy networks~\cite{PaBa14}, as well as the edge importance of complex networks~\cite{LiZh18}.

\subsection{Graph Laplacian Matrix}
Mathematically, the topological and weighted properties of a graph $\calG$ are encoded in its generalized adjacency matrix $\AA$ with the entry $a_{ij}$ denoting the adjacency relation between vertices $i$ and $j$. If vertices $i$ and $j$ are linked to each other by an edge $e$, then $a_{ij}= a_{ji}=w_{e}> 0$. Otherwise, $a_{ij}=a_{ji}=0$ indicates that vertices $i$ and $j$ are not adjacent. In a weighted graph $\calG$, the strength  $s_i$ of a vertex $i$ is defined by $s_i=\sum_{j=1}^N a_{ij}$~\cite{BaBaPaVe04}. The diagonal strength matrix of graph $\calG$ is defined to be ${\SS} = {\rm diag}(s_1, s_2, \ldots, s_N)$, and the Laplacian matrix of $\calG$ is ${\LL}={\SS}-{\AA}$.

Let $\BB \in \mathbb{R}^{|E| \times |V|}$ be the incidence matrix of $\calG$. For each edge $e$ with two end vertices $i$ and $j$, a direction is assigned arbitrarily. Let $\bb_e^\top$ be the row of matrix $\BB$ associated with edge $e$. Then the element $b_{eu}$ at row corresponding to edge $e$ and column corresponding to vertex $u$ is defined as follows: $b_{eu} = 1$ if vertex $u$ is the tail of edge $e$, $b_{eu}=-1$ if vertex $u$ is the head of  edge $e$, and $b_{eu}=0$ otherwise. Let $\ee_u$ be the $u$-th canonical basis of the space $\mathbb{R}^{|V|}$, then for an edge $e$ connecting two vertices $i$ and $j$, $\bb_e$ can also be recast as $\bb_e=\ee_{i}-\ee_{j}$.  Let  $\WW \in \mathbb{R}^{|E| \times |E|}$ be a diagonal matrix with the diagonal entry $(e,e)$ being $w_e$. Then the Laplacian matrix $\LL$ of graph $\calG$ can be written as $\LL=\BB^T\WW\BB=\sum_{e\in E}w_e\bb_e\bb_e^{\top}$. Let $\one$ denote the vector  of appropriate dimensions with all entries being ones, and let $\JJ$  be the matrix defined by $\JJ=\one\one^\top$.  Moreover, let $\mathbf{0}$ and $\mathbf{O}$ denote, respectively, the zero vector and zero matrix. Since ${\SS} \one={\AA} \one$, then ${\LL} \one =\mathbf{0}$ and ${\LL} \JJ = \JJ {\LL}=\mathbf{O}$.

\subsection{Spectrum of  Graph  Laplacian  and its Pseudoinverse}

For a connected undirected graph  $\calG=(V,E,w)$, its Laplacian matrix $\LL$ is symmetric and positive semidefinite. All  its eigenvalues  are non-negative, with a unique zero eigenvalue. Let $0=\lambda_1< \lambda_2 \leq \lambda_3\leq \dots\leq \lambda_{N-1} \leq \lambda_N$ be the $N$ eigenvalues of  $\LL$, and let $\uu_k$, $ k={1,2,\dots,N}$, be their corresponding mutually orthogonal  unit eigenvectors. Then, $\LL$ admits the following spectral decomposition:  $\LL=\sum_{k=1}^{N}\lambda_k \uu_k\uu_k^\top$, which means that the $ij$ entry $\LL_{ij}$ of $\LL$ is $\LL_{ij}=\sum_{k=2}^{N}\lambda_k \uu_{ki}\uu_{kj}$,   where $\uu_{ki}$ is the $i$th component of vector $\uu_{k}$. It can be verified~\cite{LiSc18} that among all $N$-node connected undirected weighted graphs  $\calG=(V,E,w)$, the largest eigenvalue $\lambda_{N}$ of Laplacian matrix satisfies $\lambda_{N}\leq N w_{\max}$, with equality if and only if  $\calG$ is the $N$-node complete graph, where the weight of every  edge is $w_{\max}$.

Having $0$ as an eigenvalue, $\LL$ is singular and cannot be inverted. As a substitute for the inverse we use the Moore-Penrose generalized inverse of $\LL$, that we simply call pseudoinverse of $\LL$~\cite{BeGrTh74}.  As customary,  we use $\LL^{\dagger}$ to denote its pseudoinverse, which can be written as
\begin{equation}\label{LPlus01}
\LL^{\dagger}=\sum_{k=2}^{N}\frac{1}{\lambda_k}\uu_k\uu_k^{\top}.
\end{equation}
Then, the entry $\LL^{\dagger}_{ij}$ of   $\LL^{\dagger}$ at row $i$ and column $j$ can be expressed as $\LL^{\dagger}_{ij}=\sum_{k=2}^{N}\frac{1}{\lambda_k}\uu_{ki}\uu_{kj}$. Thus, the $i$th diagonal entry $\LL^{\dagger}_{ii}$ of $\LL^{\dagger}$ is equal to $\sum_{k=2}^{N}\frac{1}{\lambda_k}\uu^2_{ki}$. Since $\LL$ is symmetric, it is the same with   $\LL^{\dagger}$~\cite{FoPiReSa07}, which  can also be seen from the fact  $\LL^{\dagger}_{ij}=\LL^{\dagger}_{ji}=\sum_{k=2}^{N}\frac{1}{\lambda_k}\uu_{ki}\uu_{kj}$.

Note that for a general symmetric matrix, it shares the same null space as its Moore-Penrose generalized inverse~\cite{BeGrTh74}.
Since ${\LL} \one =\mathbf{0}$, it turns out that ${\LL}^{\dagger} \one =\mathbf{0}$. Considering $\JJ=\one\one^\top$, we further obtain ${\LL} \JJ = \JJ {\LL}={\LL}^{\dagger} \JJ = \JJ {\LL}^{\dagger}= \mathbf{O}$. Let $\II$ be the identity matrix of approximate dimensions.  Using the spectral decompositions of  $\LL$ and ${\LL}^{\dagger}$, it is not difficult to verify that
\begin{equation*}
\left(\LL+\frac{1}{N} \JJ\right)\left(\LL^{\dagger}+\frac{1}{N} \JJ\right)= \II.
\end{equation*}
Then, it follows that
\begin{equation}\label{LPlus02}
\LL^{\dagger}=\left(\LL+\frac{1}{N} \JJ\right)^{-1}-\frac{1}{N} \JJ,
\end{equation}
which was explicitly stated  in~\cite{GhBoSa08}  and was  implicitly applied in~\cite{XiGu03,BrFl05}. Another useful consequence of the above two equalities is
\begin{equation}\label{LPlus03}
\LL \LL^{\dagger}= \LL^{\dagger} \LL=\II-\frac{1}{N} \JJ\,,
\end{equation}
which will be used in the following text.

\subsection{Interpretations of Entries of  Pseudoinverse for Graph Laplacian  }

For a graph, many key quantities are encoded in the entries of the pseudoinverse for its Laplacian matrix, which in turns provide interpretations for the  pseudoinverse $\LL^{\dagger}$ for graph Laplacian $\LL$ from different angles.

\subsubsection{Topological interpretation}

 Chebotarev and Shamis~\cite{ChSh98} offer a topological explanation for $\LL^{\dagger}$ by representing the entries $\LL^{\dagger}_{ij}$ and $\LL^{\dagger}_{ii}$ of $\LL^{\dagger}$ in terms of the weight of spanning rooted forests:
 \begin{equation}\label{EE04a}
\LL^{\dagger}_{ij}=\frac{\varepsilon(\mathcal{F}_2^{ij})-\frac{1}{N}\varepsilon(\mathcal{F}_2)}{\varepsilon(\mathcal{F}_1)}
\end{equation}
and
\begin{equation}\label{EE04b}
\LL^{\dagger}_{ii}=\frac{\varepsilon(\mathcal{F}_2^{ii})-\frac{1}{N}\varepsilon(\mathcal{F}_2)}{\varepsilon(\mathcal{F}_1)}.
\end{equation}

In~\cite{RaZh11,RaZh13} Ranjan and Zhang put forth a measure for node centrality, called topological centrality.
For a node $i$, its topological centrality is defined as $1/\LL^{\dagger}_{ii}$. Moreover, many other authors also use $\LL^{\dagger}_{ii}$ to measure node importance in different scenarios~\cite{EsHa10,VaDeCe17}.

\subsubsection{Explanations from viewpoint of electrical networks}

Various quantities  related to electrical networks are relevant to the pseudoinverse $\LL^{\dagger}$ for graph Laplacian $\LL$. 
In the context of effective resistances, resistance distance $r_{ij}$ between any pair of nodes $i$ and $j$, the  resistance distance $r_{i}$ of a node $i$, and the Kirchhoff index of the whole graph are all encoded in the entries of $\LL^{\dagger}$. First, the resistance distance $r_{ij}$ between two vertices $i$ and $j$ can be written in terms of the entries of   $\LL^{\dagger}$ as
 $r_{ij}=\LL^{\dagger}_{ii}+\LL^{\dagger}_{jj}-\LL^{\dagger}_{ij}-\LL^{\dagger}_{ji}$~\cite{KlRa93}. Second, the Kirchhoff index  $\delta$ of a graph $\calG$ with $N$ nodes is equal to $N$ times  the trace of  $\trace{\LL^{\dagger}}$  of  matrix $\LL^{\dagger}$~\cite{KlRa93}, that is,
\begin{equation}\label{Kirchhoff01}
\delta=N\,\trace{\LL^{\dagger}}=N\sum_{i=1}^{N}\LL^{\dagger}_{ii} \,.
\end{equation}
Finally, the resistance distance $r_i$ of node $i$ can also be expressed in terms of the  diagonal elements of  $\trace{\LL^{\dagger}}$~\cite{BoFr13}:
\begin{equation}\label{EE04x}
R_i=N\,\LL^{\dagger}_{ii}+ \trace{\LL^{\dagger}}=N\left(\LL^{\dagger}_{ii}+\sum_{j=1}^{N}\LL^{\dagger}_{jj}\right),
\end{equation}
which implies
\begin{equation}\label{EE04}
\LL^{\dagger}_{ii}=\frac{R_i - \trace{\LL^{\dagger}}}{N}.
\end{equation}
Thus, $\LL^{\dagger}_{ii}$ indicates how close node $i$ is within the graph. The above formula implies that the ordering of nodes by structural centrality is identical to the ordering of nodes by current flow centrality and information centrality.

In addition to the aforementioned interpretations for $\LL^{\dagger}$, there are many other representations or explanations for the entries of $\LL^{\dagger}$. For example, it was shown that the pseudoinverse $\LL^{\dagger}$ for graph Laplacian also has persuasive interpretations in  potential of nodes in an electrical network~\cite{BoFr12,BoFr13}, expected number of visit times of random walks~\cite{RaZh11,RaZh13}, and many others~\cite{KiNeSh97}.

\section{Related work}

As shown above, the pseudoinverse $\LL^{\dagger}$ for graph Laplacian  $\LL$ arises in various  application settings, and many relevant quantities can be expressed in terms of the entries of $\LL^{\dagger}$. Particularly, the  diagonal of  $\LL^{\dagger}$ contains much structural and dynamical information about a graph. For example,  the diagonal of $\LL^{\dagger}$ is sufficient to
compute the structural centrality, current flow centrality, information centrality, and Kirchhoff index of a graph. It it thus of great interest to compute  matrix  $\LL^{\dagger}$. However, computing  $\LL^{\dagger}$ is a theoretical challenge.   By virtue of~\eqref{LPlus01} and~\eqref{LPlus02}, na\"{\i}ve  methods for computing $\LL^{\dagger}$ of a graph involve either calculating the eigenvalues and eigenvector of  $\LL$  or  inverting a suitable perturbed version $\LL+\frac{1}{n} \JJ$ of $\LL$ and then subtracting the perturbation $\frac{1}{n} \JJ$. Both straightforward ways to compute $\LL^{\dagger}$  cost $O(N^3)$ time, which are infeasible for large-scale networks with millions of nodes and edges.

In order to speed up the computation of matrix $\LL^{\dagger}$, a lot of endeavors have been devoted to fast algorithms for evaluating $\LL^{\dagger}$.  In~\cite{RaZhBo14}, an incremental approach was designed to compute $\LL^{\dagger}$. For general cases, this method performs better than the common approaches, but for the worst case it still has a computation cost of $O(N^3)$. Moreover, several approximation algorithms were presented to estimate matrix $\LL^{\dagger}$~\cite{BoFr12} or the diagonal entries of a more general matrix~\cite{TaSa11,TaSa12,WuLaKaStGa16}, with an aim to reduce the computational complexity. Although these algorithms have low complexity, they have no approximation guarantee. Finally, many algorithms were developed to address a related issue of estimating the trace of a matrix~\cite{AvTo11,GaStOr17,UbChSa17}. However, these algorithms do not apply to evaluate the diagonal of $\LL^{\dagger}$. It is thus desirable to propose an effective and efficient algorithm for approximating $\LL^{\dagger}$ that provides error bounds on the diagonal entries. This is the main research subject of the present paper.

In order to obtain an exact expression for every diagonal element $\LL_{ii}^{\dagger}$, $i={1,2,\dots,N}$, of the pseudoinverse of $\LL$ for a graph with $N$ vertices,  one may determine all $N-1$ non-zero eigenvalues of $\LL$ and their corresponding mutually orthogonal unit eigenvectors. Making use of the approach similar to that in~\cite{YiYaZhZhPa22}, it is not difficult to derive a rigorous formular for each $\LL_{ii}^{\dagger}$ for some particular graphs, including the path graph, the ring graph, the star graph, and the complete graph, since the eigenvalues and eigenvectors of their Laplacian matrices can be obtained explicitly. However, these graphs cannot mimic real networks, most of which display the striking scale-free small-world behaviors~\cite{Ne03}. Thus far, rigorous solution for $\LL_{ii}^{\dagger}$ associated with scale-free small-world networks is still missing. One the other hand, some real networks (e.g., power grid) are exponential with their degree distribution decaying exponentially~\cite{AmScBaSt00}, but related analytical work for $\LL_{ii}^{\dagger}$ is also much less.

\section{Fast Algorithm for  Approximating the Diagonal of Pseudoinverse of Graph Laplacian }


In this section, we propose an algorithm to compute an approximation of all diagonal entries of  $\LL^{\dagger}$ in nearly linear time with respect to the number of edges. Our algorithm has an error guarantee  with a high probability.

To achieve our goal, we first reduce the problem for computing a diagonal entry of  $\LL^{\dagger}$ to calculating the $\ell_2$  norm of a vector. Then  using the techniques of  random projections and linear system solvers, we estimate the $\ell_2$ norm,  aiming  to reduce the computation complexity. We now express  $\LL^{\dagger}_{uu}$ in an Euclidian norm. According to~\eqref{LPlus03} and the relation $\LL \JJ=\mathbf{O}$, we obtain
\begin{align*}
&\LL^{\dagger}_{uu} =
\ee_u^T \LL^\dag \LL \LL^\dag \ee_u
=
\ee_u^T \LL^\dag \BB^T \WW \BB \LL^\dag \ee_u \\
= &\ee_u^T \LL^\dag \BB^T \WW^{1/2}
\WW^{1/2} \BB \LL^\dag \ee_u = \norm{\WW^{1/2} \BB \LL^\dag \ee_u}^2.
\end{align*}

In this way, we have reduced the computation of a diagonal element $\LL^{\dagger}_{uu}$ of  $\LL^{\dagger}$ to evaluating the $\ell_2$ norm of  a vector  in $\mathbb{R}^M$. However, through using this  $\ell_2$ norm, the complexity for exactly computing all diagonal elements is still very high. Fortunately,  applying the Johnson-Lindenstrauss lemma~\cite{JoLi84,Ac01,Ac03},  the $\ell_2$ norm can
be nearly  preserved by projecting the vector onto a low-dimensional subspace, while the computational cost  is significantly reduced. For consistency, we  introduce the Johnson-Lindenstrauss lemma~\cite{JoLi84,Ac01,Ac03}.
\begin{lemma}
	\label{lemma:JL}
	Given fixed vectors $\vv_1,\vv_2,\ldots,\vv_N\in \mathbb{R}^M$ and
	$\epsilon>0$, let
	$\QQ_{k\times M}$ be a random $\pm 1/\sqrt{k}$ matrix (i.e., independent Bernoulli entries)
	with $k\ge 24\log N/\epsilon^2$. Then with probability at least $1-1/N$,
	\[(1-\epsilon)\|\vv_i-\vv_j\|_2^2\le \|\QQ \vv_i-\QQ \vv_j\|_2^2\le
	(1+\epsilon)\|\vv_i-\vv_j\|_2^2\] for all pairs $i,j\le N$.
\end{lemma}

Let $\QQ$ be a $k\times M$ random projection matrix. By Lemma~\ref{lemma:JL},  $\|\QQ \WW^{\frac{1}{2}} \BB \LL^{\dagger}\ee_u\|$ is a good approximation for $\|\WW^{\frac{1}{2}} \BB \LL^{\dagger} \ee_u\|$. Here we can use sparse matrix multiplication to compute $\QQ \WW^{\frac{1}{2}} \BB$, which takes $2M \times 24\log N/\epsilon^2+M$ time, since $\BB$ has $2M$ non-zero elements and $\WW^{\frac{1}{2}}$ is a diagonal matrix. However, computing $\ZZ=\QQ \WW^{\frac{1}{2}} \BB \LL^{\dagger}$ directly involves calculating $\LL^{\dagger}$.  In order to avoid calculating $\LL^{\dagger}$, we resort to the nearly-linear time solver~\cite{SpTe04,SpTe14} as stated in the following lemma. In the sequel, we use the notation $\Otil(\cdot)$ to hide $\mathrm{poly}(\log N )$ factors.
\begin{lemma}
\label{lemma:ST}
There is an algorithm $\xx = \mathtt{LaplSolve}(L,\yy,\theta)$ which	takes a Laplacian matrix $\LL$,
a column vector $\yy$, and an error parameter $\theta > 0$, and returns a column vector $\xx$ satisfying  $\one^{\top} \xx=0$ and
\[
\|\xx - \LL^{\dagger} \yy\|_{\LL} \leq \theta \|\LL^{\dagger} \yy\|_{\LL},
\]
where $\norm{\yy}_{\LL} = \sqrt{\yy^{\top} \LL \yy}$.
The algorithm runs in expected time $\Otil \left(M \log(1/\theta) \right)$.
\end{lemma}
Using Lemma~\ref{lemma:ST} we can avoid  calculating $\LL^{\dagger}$ by solving the system of equations $\LL \zz_i=\qq_i$, $i=1,\ldots,k$, where  $\zz^\top_i$ and $\qq^\top_i$ are the $i$th row of $\ZZ$ and $\QQ \WW^{\frac{1}{2}} \BB$, respectively.  Lemma~\ref{lemma:ST} indicates that $\zz^\top _i$ can be efficiently approximated by using $\mathtt{LaplSolve}$.

Lemmas~\ref{lemma:JL} and~\ref{lemma:ST} are critical to proving the error bounds  of our algorithm, which also involves Frobenius norm of a matrix. For a matrix $\AA \in \mathbb{R}^{M \times N}$ with entries $a_{ij}$ ($i=1,2,\cdots,M$, $j=1,2,\cdots,N$), its Frobenius norm $\norm{\AA}_{F}$ is defined as
\begin{equation*}
\norm{\AA}_{F}:=\sqrt{\sum_{i=1}^{M} \sum_{j=1}^{N} a_{ij}^2} = \sqrt{\trace{\AA^{\top} \AA}}.
\end{equation*}
By definition, it is easy to verify that $\norm{\WW}_{F} \geq w_{\min}$.

\begin{lemma}\label{lem:error2}
	Given an undirected weighted graph $\calG=(V, E,w)$ with $N$ nodes and Laplacian matrix $\LL$,  an approximate factor $0<\epsilon <1/2$,  and a $k\times N$ matrix $\ZZ$,  satisfing the following relation
	\[
	(1-\epsilon) \LL^{\dagger}_{uu}
	\leq
	\|\ZZ \ee_u\|^{2}
	\leq
	(1+\epsilon) \LL^{\dagger}_{uu},
	\]
	for an arbitrary node $u\in V$ and
	\begin{equation*}
	\begin{split}
	(1-\epsilon) \|\WW^{\frac{1}{2}} \BB \LL^{\dagger} (\ee_u-\ee_v)\|^2
	\leq
	\|\ZZ (\ee_u - \ee_v)\|^2& \\
	\leq
	(1+\epsilon) \|\WW^{\frac{1}{2}} \BB \LL^{\dagger} (\ee_u-\ee_v)\|^2&
	\end{split}
	\end{equation*}
for any node pair  $u,v \in V$,
let $\zz^\top_i$ be the $i$-th row of the $k \times N$ matrix $\ZZ$, and let $\tilde{\zz}^\top_i$ be an approximation of  $\zz_i$ for all $i\in \{1,2,...,k\}$ obeying
	\begin{equation}\label{EE18} \|\zz_i-\tilde{\zz}_i\|_{\LL}\le\delta
	\|\zz_{i}\|_{\LL},
	\end{equation}
	where
	\begin{equation}\label{EE19}
	\delta \leq  \frac{\epsilon }{3}
	 \sqrt{\frac{(N-1)(1-\eps)w_{\min}}{N^4(1+\eps)w_{\max}}}.
	\end{equation}
	Then for an arbitrary node $u \in V$,
		\begin{align}
		\label{EE20}
		(1 - \epsilon)^2 \LL^{\dagger}_{uu}
		\leq
		\|\ZZtil \ee_{u}\|^2
		\leq
		(1 + \epsilon)^2 \LL^{\dagger}_{uu},
		\end{align}
where $\ZZtil^\top = [\tilde{\zz}_1, \tilde{\zz}_2, ..., \tilde{\zz}_k]$.
\end{lemma}
\begin{proof}
We first show that in order to prove (\ref{EE20}), it suffices  to show that for any node $u\in V$,
	\begin{align}\label{EE211}
		&\quad	\abs{\|\ZZ \ee_u\|^2-\|\ZZtil \ee_u\|^2}\notag\\
		&=
		\abs{\|\ZZ \ee_u\|-\|\ZZtil \ee_u\|}\times	\abs{\|\ZZ \ee_u\|+\|\ZZtil \ee_u\|} \notag\\
		&\le
		\left(\frac{2\epsilon}{3}+\frac{\epsilon^2}{9}\right)\norm{\ZZ \ee_u}^2,
		\end{align}	
which is satisfied if the following relation holds
		\begin{equation}\label{EE21}
		\abs{\|\ZZ \ee_u\|-\|\ZZtil \ee_u\|} \le
		\frac{\epsilon}{3}\norm{\ZZ \ee_u}.
		\end{equation}
This can be  explained by the following arguments. On the one hand, if $\abs{\norm{\ZZ\ee_u}^2-\|{\ZZtil\ee_u}\|^2} \le \left(\frac{2\epsilon}{3}+\frac{\epsilon^2}{9}\right)\norm{\ZZ\ee_u}^2$, it follows that
\begin{align*}
\left(1-\frac{2\epsilon}{3}-\frac{\epsilon^2}{9}\right)\norm{\ZZ\ee_u}^2 \le & \|{\ZZtil\ee_u}\|^2  \\
\le & \left(1+\frac{2\epsilon}{3}+\frac{\epsilon^2}{9}\right)\norm{\ZZ\ee_u}^2.
\end{align*}
Since $0<\epsilon \le 1/2$ (meaning $\epsilon^2/9 \leq \epsilon/3$) and  $(1-\epsilon) \LL^{\dagger}_{uu} \leq\|\ZZ \ee_{u} \|^{2}\leq(1+\epsilon) \LL^{\dagger}_{uu}$, the above formula directly leads to~\eqref{EE20}.
On the other hand, if~\eqref{EE21} holds, then one has $\|\ZZtil\ee_u\| \le (1+\frac{\epsilon}{3})\norm{\ZZ\ee_u}$. In other words,
\begin{equation*}
\abs{\norm{\ZZ\ee_u}+\|\ZZtil\ee_u\|} \le \left(2+\frac{\epsilon}{3}\right)\norm{\ZZ\ee_u},
\end{equation*}
which leads to~\eqref{EE211}.	

We now prove that~\eqref{EE21} is true. Considering $\ZZ \one = \ZZtil \one=\mathbf{0}$ and making use of  the triangle inequality twice, we obtain
		\begin{align*}
		&\abs{\norm{\ZZ \ee_u} - \norm{\ZZtil \ee_u}}
		\leq \norm{(\ZZ - \ZZtil) \ee_u} 
		\\
		= &
		\norm{(\ZZ - \ZZtil) \kh{\ee_u - \frac{1}{N}\one}} \\
		= & N^{-1}
		\norm{(\ZZ - \ZZtil)
			\sum\nolimits_{v\neq u} \kh{\ee_u - \ee_v}} \\
		\leq &
		N^{-1} \sum\limits_{v\neq u}
		\norm{(\ZZ - \ZZtil) \kh{\ee_u - \ee_v}} .
		\end{align*}
Let $P_{uv}$ denote a  simple path linking vertices $u$ and $v$. Then, 
		\begin{align*}
		&N^{-1} \sum\limits_{v\neq u}
		\norm{(\ZZ - \ZZtil) \kh{\ee_u - \ee_v}}
		\\
		\leq &
		N^{-1} \sum\limits_{v\neq u}
		\sum\limits_{a\sim b \in P_{uv}}
		\norm{(\ZZ - \ZZtil) \kh{\ee_a - \ee_b}} \\
		\leq & \kh{\sum\limits_{v\neq u}
			\sum\limits_{a\sim b \in P_{uv}}
			\norm{(\ZZ - \ZZtil) \kh{\ee_a - \ee_b}}^2 }^{1/2} \\
		\leq & N^{1/2} \kh{\sum\limits_{a\sim b\in E}
			\norm{(\ZZ - \ZZtil) \kh{\ee_a - \ee_b}}^2}^{1/2} \\
		= & N^{1/2} \norm{(\ZZ - \ZZtil) \BB^\top }_F,
		\end{align*}
where the second inequality and the third inequality are derived based on the triangle inequality and Cauchy-Schwarz inequality, respectively. The term $N^{1/2} ||(\ZZ - \ZZtil) \BB^\top||_F$ is evaluated as
		\begin{align*}
		&N^{1/2} \norm{(\ZZ - \ZZtil) \BB^T }_F\\
		\leq &
		N^{1/2} w_{\mathrm{min}}^{-1/2}
		\norm{(\ZZ - \ZZtil) \BB^T \WW^{1/2}}_F \\
		\leq &
		N^{1/2} w_{\mathrm{min}}^{-1/2} \theta
		\norm{\ZZ \BB^T \WW^{1/2}}_F
		\qquad \\
		= & N^{1/2} w_{\mathrm{min}}^{-1/2} \theta
		\kh{
			\sum_{a\sim b \in E} w_{a\sim b}
			\norm{\ZZ (\ee_a - \ee_b)}^2}^{1/2},
		\end{align*}
	where the first inequality is obtained according to the relation $\norm{\WW}_{F} \geq w_{\mathrm{min}}$, and the second inequality is obtained according to~\eqref{EE18}, with $\theta$ being given by~\eqref{EE19}. Applying the relation $\norm{\ZZ (\ee_a - \ee_b)}^2\leq (1+\epsilon) r_{ab}$ and Lemma~\ref{Foster}, $N^{1/2} ||(\ZZ - \ZZtil) \BB^\top||_F$ is further bounded as
\begin{align*}
& N^{1/2} \norm{(\ZZ - \ZZtil) \BB^T }_F\\
\leq & N^{1/2} w_{\mathrm{min}}^{-1/2} \theta
\kh{
	(1 + \eps)
	\sum_{a\sim b \in E} w_{a\sim b} r_{ab}  }^{1/2}
\\
= &
N^{1/2} \theta \kh{\frac{(1 + \eps)(N-1)}{w_{\mathrm{min}}}}^{1/2}.
\end{align*}	

We continue to provide a lower bound of $\norm{\ZZ \ee_u}^2$ as
	\begin{equation}\label{EE22}
	\begin{split}
	&\quad\norm{\ZZ \ee_u}^2
	\geq (1-\epsilon)\ee_u^{\top} \LL^{\dagger} \ee_u \\
	&= (1-\epsilon)\kh{\ee_u - \frac{1}{N}\one}^\top \LL^{\dagger} \kh{\ee_u - \frac{1}{N}\one}\\
	&\geq (1-\epsilon) (\lambda_{N})^{-1} \frac{N-1}{N}
	\geq (1-\epsilon)\frac{(N-1)^2}{N^3 w_{\max}}
	\end{split}
	\end{equation}
	In~\eqref{EE22}, the inequalities are obtained due to the following arguments.  Since $\one$ is  an eigenvector of $\LL^{\dagger}$ corresponding to the unique eigenvalue  $0$, and   $\ee_u-({1}/{N})\one$ is orthogonal to  vector $\one$ obeying $\norm{\ee_u-\frac{1}{N}\one}^2=(N-1)^2/N^2$, then we have $(\ee_u-\frac{1}{N}\one)^{\top} \LL^{\dagger} (\ee_u-\frac{1}{N}\one) \ge \lambda_{N}^{-1} \norm{\ee_u-\frac{1}{N}\one}^2$. Combining the above-obtained results and the $\delta$ value given by (\ref{EE19}), it follows that
	\begin{small}
		\begin{align*}
		&\quad \frac{
			\abs{ \norm{\ZZ \ee_u} -  \norm{\ZZtil  \ee_u}}
		}{
			\norm{\ZZ \ee_u}
		}\\
		&\le
		\theta \kh{\frac{(1 + \eps)N(N-1)}{w_{\mathrm{min}}}}^{1/2}\left(\frac{N^3
			w_{\max}}{(N-1)^2(1-\epsilon)}\right)^{1/2}
			\le \frac{\epsilon}{3},
		\end{align*}
	\end{small}
	which is equivalent to~\eqref{EE21} and thus finishes the proof.
\end{proof}

Lemma~\ref{lem:error2} leads to the following theorem.

\begin{theorem}
	\label{TheoAlg2}
There is a $\Otil(M\log{c}/\epsilon^2)$ time algorithm, which  inputs $0<\epsilon<1$ and $\calG=(V,E,w)$ where $c=\frac{w_{\max}}{w_{\min}}$, and returns a $(24\log N/\epsilon^2)\times N$ matrix $\ZZtil$ such that with probability at least $1-1/N$,
	\begin{align}
	(1-\epsilon)^2 \LL^{\dagger}_{uu} \leq \|\ZZtil \ee_u\|^2 \leq (1+\epsilon)^2 \LL^{\dagger}_{uu}\nonumber
	\end{align}
	for any node $u \in V$.
\end{theorem}

Based on Theorem~\ref{TheoAlg2} and~\eqref{Kirchhoff01}, we provide a randomized algorithm $\mathtt{Approx\mathcal{D}}$ to approximately compute $\LL^{\dagger}_{uu}$ for all  nodes  $u \in V$ and the Kirchhoff index $\theta$ for a general graph  $\calG=(V,E,w)$, the pseudocode of which is shown in Algorithm~\ref{ALG02}.


\begin{algorithm2e}\label{ALG02}
	\caption{$\text{$\mathtt{Approx\mathcal{D}}$}(\calG, \epsilon)$}
	\Input{
		$\calG$: a connected undirected weighted graph. \\
		$\epsilon$: an approximation parameter 
	}
	\Output{
		$\tilde{\calL^{\dagger}}=\{u,\tilde{\LL}^{\dagger}_{uu}| u \in V\}$: $\tilde{\LL}^{\dagger}_{uu}$ is an approximation of the diagonal element $\LL^{\dagger}_{uu}$ associated with vertex $u$;
		$\tilde{\delta}$: approximation of the Kirchhoff index $\delta$
	}
	$\LL=$ Laplacian matrix of $\calG$,\,\\
	Construct a matrix $\QQ_{k \times M}$,  where $k=\lceil 24\log N/\epsilon^2 \rceil$ and each entry of $\QQ_{k \times M}$ is $\pm 1/\sqrt{k}$ with identical probability\;
	\For{$i=1$ to $k$}{
		 $\qq_i^\top $=the $i$-th row of  $\QQ_{k \times M} \WW^{1/2} \BB$  \\
		$\tilde{\zz}_i=\mathtt{LaplSolve}(\LL, \qq_i, \delta)$ where parameter $\delta$ is given by~(\ref{EE19}) \\	
	}
	\For{each $u\in V$}{
		$\tilde{\LL}^{\dagger}_{uu} = \|\ZZtil_{:,u}\|^2$
	}
	$\tilde{\delta}=N\sum_{u \in V} \tilde{\LL}^{\dagger}_{uu}$\;
	\Return $\tilde{\calL^{\dagger}}=\{u,\tilde{\LL}^{\dagger}_{uu}| u \in V\}$ and $\tilde{\delta}$
\end{algorithm2e}


\section{Experiment Results on Real-World Networks}

In this section, we experimentally evaluate the efficiency and accuracy of our approximation algorithm on real networks. Here we only consider the diagonal entries of pseudoinverse for Laplacian matrix, excluding the Kirchhoff index, since this is enough to reach our goal.

We evaluate the algorithm $\mathtt{Approx\mathcal{D}}$ on a large set of  real-world networks from different domains.  The  data of these networks are taken from the Koblenz Network Collection~\cite{Ku13}. We run our experiments on the largest connected components (LCC) of these networks, related information of which  is shown  in Table~\ref{tab:network_size}.

\begin{table}
		\centering
		\normalsize
		\tabcolsep=8pt
		\fontsize{7}{8}\selectfont
			\caption{Statistics of datasets. For a network with $N$ nodes and $M$ edges, we denote the number of nodes and edges in its largest connected component by $N'$ and $M'$, respectively.}
			\label{tab:network_size}
		\begin{threeparttable}
            \begin{tabularx}{\linewidth}{m{1.5cm}<{\centering}m{1.1cm}<{\centering}m{1.1cm}<{\centering}m{1.1cm}<{\centering}m{1.1cm}<{\centering}}
				\Xhline{2\arrayrulewidth}
				\specialrule{0em}{1.5pt}{1pt}
				Network & $N$ & $M$ & $N'$ & $M'$ \cr
				\midrule
				Jazz musicians      & 198  & 2,742 & 198 & 2,742\cr
				Chicago             & 1,467 & 1,298 & 823 & 822 \cr
				Hamster full	    & 2,426 & 16,631 & 2,000 & 16,098 \cr
				Facebook            & 4,039 & 88,234 & 4,039 & 88,234\cr
				CA-GrQc             & 5,242 & 14,496 & 4,158 & 13,422 \cr
				Reactome            & 6,327 & 147,547 & 5,973 & 145,778 \cr
				Route views         & 6,474 & 13,895  & 6,474 & 12,572 \cr
				PGP                 & 10,680 & 24,316 & 10,680 & 24,316 \cr 
				CA-HepPh      	    & 12,008 & 118,521 & 11,204 & 117,619 \cr
				Astro-ph      	    & 18,772 & 198,110 & 17,903 & 196,972 \cr
				CAIDA               & 26,475 & 53,381 & 26,475 & 53,381\cr
				Brightkite          & 58,228 & 214,078 & 56,739 & 212,945 \cr
				Livemocha*          & 104,103 & 2,193,083 & 104,103& 2,193,083\cr
				WordNet*            & 146,005 & 656,999 & 145,145 & 656,230 \cr
				Gowalla*            & 196,591 & 950,327 & 196,591& 950,327 \cr
				com-DBLP*           & 317,080 & 1,049,866 & 317,080 & 1,049,866 \cr
				Amazon*             & 334,863 & 925,872 & 334,863 & 925,872 \cr
				Pennsylvania*       & 1,088,092 & 1,541,898 & 1,087,562 & 1,541,514 \cr
				roadNet-TX*         & 1,379,917 & 1,921,660 & 1,351,137 & 1,879,201 \cr
				\Xhline{2\arrayrulewidth}
			\end{tabularx}
		\end{threeparttable}
	\end{table}

We run all the experiments on a Linux box with an \textit{Intel i7-7700K @ 4.2-GHz (4 Cores)} and with \textit{32GB} memory. We implement the algorithm $\mathtt{Approx\mathcal{D}}$ in \textit{Julia v0.6.0}, where  the $\mathtt{LaplSolve}$ is from~\cite{RaSu16}, the \textit{Julia language} implementation of which is accessible on the website\footnote{http://danspielman.github.io/Laplacians.jl/latest/}.

\begin{table}
	\centering
	\tabcolsep=3.2pt
	\small
	\fontsize{7}{8}\selectfont
	\caption{The running time (seconds, $s$) of $\mathtt{Exact\mathcal{D}}$ and $\mathtt{Approx\mathcal{D}}$ with various $\epsilon$ on  real-world networks}
		\label{tab:runtime_comparison}
	\begin{threeparttable}
		\begin{tabular}{cccccccc}
			\Xhline{2\arrayrulewidth}
			\multirow{2}{*}{Network}&
			\multirow{2}{*}{$\mathtt{Exact\mathcal{D}}$}&
			\multicolumn{6}{c}{$\mathtt{Approx\mathcal{D}}$ ($s$) with various $\epsilon$}\cr
			\cmidrule(lr){3-8}
			&  ($s$) & $0.3$ & $0.25$ & $0.2$ & $0.15$ & $0.1$ & $0.05$\cr
			\midrule
			Jazz musicians       & 0.001  & 0.019 & 0.027 & 0.041 & 0.068 & 0.148 & 0.571 \cr
			Chicago              & 0.03   & 0.007 & 0.009 & 0.013 & 0.023 & 0.05  & 0.193\cr
			Hamster full	     & 0.372  & 0.200 & 0.308 & 0.435 & 0.749 & 1.576 & 6.271\cr
			Facebook	     & 2.758  & 1.060 & 1.289 & 2.119 & 3.675 & 7.530 & 30.81 \cr
			CA-GrQc              & 2.997  & 0.336 & 0.386 & 0.571 & 1.174 & 2.336 & 9.410\cr
			Reactome             & 8.598  & 1.876 & 2.443 & 3.700 & 6.424 & 13.45 & 56.13\cr
			Route views          & 10.88  & 0.261 & 0.333 & 0.486 & 0.932 & 2.038 & 7.776\cr
			PGP                  & 46.79  & 0.751 & 0.917 & 1.545 & 2.551 & 5.991 & 22.20\cr
			CA-HepPh      	     & 53.97  & 1.952 & 2.508 & 3.540 & 7.377 & 14.64 & 58.65\cr
			Astro-ph      	     & 216.6  & 4.803 & 5.118 & 7.788 & 14.52 & 28.89 & 129.1\cr
			CAIDA                & 700.8  & 1.703 & 1.746 & 2.723 & 5.118 & 10.40 & 42.16\cr
			Brightkite           & 4415   & 9.241 & 8.268 & 12.61 & 21.88 & 51.51 & 212.8\cr
			Livemocha*           & --     & 62.16 & 79.85 & 130.5 & 207.8 & 482.2 & 1842\cr
			WordNet*             & --     & 23.96 & 32.36 & 46.87 & 82.77 & 183.5 & 793.1\cr
			Gowalla*             & --     & 36.66 & 49.44 & 79.38 & 131.3 & 296.8 & 1241\cr
			com-DBLP*            & --     & 60.97 & 89.68 & 126.9 & 240.3 & 522.2 & 2090\cr
			Amazon*              & --     & 80.01 & 111.0 & 177.5 & 279.9 & 694.3 & 2604\cr
			Pennsylvania*        & --     & 301.7 & 443.2 & 672.8 & 1186 & 2823 & 10402\cr
			roadNet-TX*          & --     & 405.3 & 548.1 & 883.4 & 1658 & 3478 & 14160\cr
			\Xhline{2\arrayrulewidth}
		\end{tabular}
		\end{threeparttable}
\end{table}


To demonstrate the efficiency of our approximation algorithm $\mathtt{Approx\mathcal{D}}$, in Table \ref{tab:runtime_comparison}, we compare the running time of $\mathtt{Approx\mathcal{D}}$ with that of the accurate algorithm called $\mathtt{Exact\mathcal{D}}$ that calculates the diagonal elements of $\LL^{\dagger}$ by  using \eqref{LPlus02}.  
The results show that for  moderate $\epsilon$, $\mathtt{Approx\mathcal{D}}$ is significantly  faster than  $\mathtt{Exact\mathcal{D}}$, especially for large networks.   For the last seven networks with node number ranging from $10^5$ to $10^6$, we cannot run the $\mathtt{Exact\mathcal{D}}$ algorithm due to memory limit and high time cost. In contrast, for these networks, we can approximately compute all diagonal entries of $\LL^{\dagger}$. This further show that   $\mathtt{Approx\mathcal{D}}$ is efficient and scalable, which is suitable for large networks.

In addition to the efficiency, we also evaluate the accuracy of  the approximation algorithm $\mathtt{Approx\mathcal{D}}$. To this end, we compare the approximate results of $\mathtt{Approx\mathcal{D}}$ with the exact results calculated by  $\mathtt{Exact\mathcal{D}}$. In Table~\ref{tab:accuracy}, we provide the mean relative error $\sigma$ of our approximation algorithm, where $\sigma$ is defined as $\sigma=\frac{1}{N}\sum_{u\in V}|\LL^{\dagger}_{uu}-\tilde{\LL}^{\dagger}_{uu}|/{\LL^{\dagger}_{uu}}$. The results show that the actual mean relative errors for all $\epsilon$ and all networks are insignificant, which are magnitudes smaller than the theoretical guarantee. Thus, the  approximation algorithm  $\mathtt{Approx\mathcal{D}}$ leads to very accurate results in practice.

\begin{table}
	\centering
	\small
	\tabcolsep=2.8pt
	\fontsize{7}{6.8}\selectfont
	\begin{threeparttable}
		\caption{Mean relative errors $\sigma$ of $\mathtt{Approx\mathcal{D}}$ with various $\eps$.}
		\label{tab:accuracy}
		\begin{tabular}{ccccc}
			\toprule
			\multirow{2}{*}{network}&
			\multicolumn{4}{c}{Mean relative error $\sigma$ with various $\epsilon$}\cr
			\cmidrule(lr){2-5}
			& $0.3$ & $0.2$ & $0.1$ & $0.05$\cr
			\midrule
			Jazz musicians &$\num{1.15e-1}$  & $\num{8.51e-2}$  & $\num{4.18e-2}$ & $\num{1.82e-2}$\cr
			Chicago        &$\num{7.90e-2}$  & $\num{5.37e-2}$  & $\num{2.87e-2}$ & $\num{1.42e-2}$\cr
			Hamster full   &$\num{8.77e-2}$  & $\num{5.94e-2}$  & $\num{2.97e-2}$ & $\num{1.49e-2}$\cr
			Facebook       &$\num{9.15e-2}$  & $\num{5.81e-2}$  & $\num{2.91e-2}$ & $\num{1.50e-2}$\cr
			CA-GrQc        &$\num{7.98e-2}$  & $\num{5.43e-2}$  & $\num{2.64e-2}$ & $\num{1.33e-2}$\cr
			Reactome       &$\num{8.61e-2}$  & $\num{5.64e-2}$  & $\num{2.86e-2}$ & $\num{1.41e-2}$\cr
			Route views    &$\num{5.78e-2}$  & $\num{3.83e-2}$  & $\num{1.95e-2}$ & $\num{9.71e-3}$\cr
			PGP            &$\num{7.13e-2}$  & $\num{4.71e-2}$  & $\num{2.35e-2}$ & $\num{1.17e-2}$\cr
			CA-HepPh       &$\num{7.73e-2}$  & $\num{5.11e-2}$  & $\num{2.57e-2}$ & $\num{1.29e-2}$\cr
			Astro-ph       &$\num{7.93e-2}$  & $\num{5.26e-2}$  & $\num{2.65e-2}$ & $\num{1.31e-2}$\cr
			CAIDA          &$\num{5.31e-2}$  & $\num{3.53e-2}$  & $\num{1.75e-2}$ & $\num{8.79e-3}$\cr
			Brightkite     &$\num{5.94e-2}$  & $\num{3.93e-2}$  & $\num{1.97e-2}$ & $\num{9.82e-3}$\cr
			\bottomrule
		\end{tabular}
	\end{threeparttable}
\end{table}

\section{Exact Formulas for  Diagonal of Pseudoinverse of Laplacian of Model Networks }
\label{graphs.sec}

For a general graph,  exact expression for  the diagonal  of the pseudoinverse $\LL^{\dagger}$ for its Laplacian is very difficult to obtain. However, for some model networks constructed by an iterative way,   the diagonal of $\LL^{\dagger}$  can be explicitly determined. For example, for the Koch networks~\cite{ZhZhXiChLiGu09} and uniform recursive tree~\cite{LiZh13},  we can derive  exact expressions for this  relevant quantity. In this section, we determine the expressions for  the diagonal of $\LL^{\dagger}$ for Koch networks and uniform recursive tree by exploiting \eqref{EE04}, while in the  next section, we will we use the these results to further evaluate the performance of our approximation algorithm $\mathtt{Approx\mathcal{D}}$.

\subsection{Diagonal of Pseudoinverse of Laplacian for Koch Networks }
\label{graphs.secKoch}

In the subsection, we derive the formula for diagonal elements of pseudoinverse for Laplacian of Koch Networks.

\subsubsection{Network construction and properties}
\label{built.secCon}

The Koch networks  translated from the Koch curves~\cite{Sc65} that  can be applied to design fractal antenna~\cite{BaRoCa00}, are constructed in an iterative way~\cite{ZhZhXiChLiGu09}. Let $\mathcal{K}_{g}$ ($g \geq 0$) denote the Koch networks after $g$ iterations. Initially ($g=0$), $\mathcal{K}_{0}$ consists of a triangle with three
nodes and three edges. For $g\geq 1$, $\mathcal{K}_{g}$ is obtained from $\mathcal{K}_{g-1}$ by
performing the following operations. For each of the three nodes in
every existing triangle in $\mathcal{K}_{g-1}$, we generate two new nodes,   which and their
``mother'' nodes are connected to one another forming a new triangle. Figure~\ref{KochNet} illustrates the growth process of the Koch networks.

\begin{figure}
\centering
\includegraphics[width=2in]{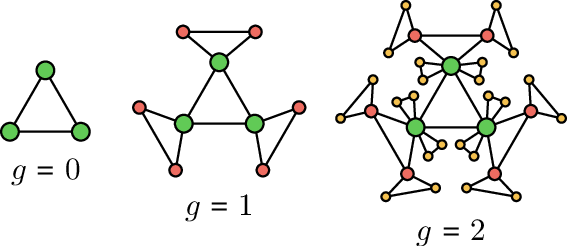}
\caption{Construction process for the Koch network.}
\label{KochNet}
\end{figure}

Let $L_v(g_i)$ be the number of nodes  generated at iteration $g_i$ ($g_i \ge 1$). By construction, $L_v(g_i)=6\times4^{g_i-1}$. In $\mathcal{K}_{g}$, the initial three nodes have the largest degree, we thus call them the hub nodes. Let $N_g$ and $M_g$ denote, respectively, the numbers of nodes and edges  in  $\mathcal{K}_{g}$. It is easy to derive that $N_g=2\times 4^{g}+1$ and  $e_g=3\times 4^{g}$. Then, the average degree of $\mathcal{K}_{g}$ is $(6\times 4^{g})/(2\times
4^{g}+1)$, which is asymptotically equal to $3$ for large $g$.


The Koch networks display some remarkable characteristics as observed in a large variety of real-world systems~\cite{Ne03}. It is scale-free~\cite{BaAl99}, since  the degree distribution $P(k)$ follows a power-law behavior $P(k) \sim k^{-3}$. In addition, the Koch networks   exhibit the small-world effect~\cite{WaSt98}, since the average geodesic distance grows logarithmically with $N_g$ and the average clustering coefficient is high, converging to a nonzero value $0.82$ for large $g$.

\subsubsection{Explicit expression for diagonal of pseudoinverse of  laplacian}
\label{sec:1}

Let $\LL^{\dagger}(g)$ denote the Laplacian matrix of network $\mathcal{K}_{g}$. Let $V_g$ denote the set of the $N_g$ vertices in the Koch network $\mathcal{K}_g$. Let $\LL_{xx}^{\dagger}(g)$ be the diagonal entry   of $\LL^{\dagger}(g)$ associated with node $x$.  Let $r_{xy}(g)$ denote the resistance distance between two vertices $x$ and $y$ in $\mathcal{K}_{g}$.  Then the  resistance distance $r_x(g)$ of node $x$ in $\mathcal{K}_{g}$ is  $r_x(g)=\sum_{y\in V_g}r_{xy}(g)$. And let $\delta(g)$ denote the  Kirchhoff index of  $\mathcal{K}_{g}$.

By~\eqref{EE04}, in order to determine $\LL_{xx}^{\dagger}(g)$, we need first know $\delta(g)$ and  $r_x(g)$.  It was shown~\cite{WuZhCh12} that the Kirchhoff index of  $\mathcal{K}_{g}$ is
\begin{equation}
\label{eqa:Kichoff}
\delta(g)=\frac{2^{4g+1} (6 g+7)+4^{g+1}}{9}\,.
\end{equation}

We now calculate $r_{x}(g)$.  Let $d_{xy}(g)$ be the shortest-path distance between a pair of nodes $x$ and $y$ in $\mathcal{K}_g$. By construction of Koch networks, we have $r_{xy}(g)=2d_{xy}(g)/3$. Define
$d_x(g)=\sum_{y \in V_g} d_{xy}(g)$ as the shortest-path distance of node $x$, which
is the sum of the shortest-path distances from  $x$ to all other nodes in $\mathcal{K}_g$.
Then,
\begin{equation}\label{Hlfxg}
r_{x}(g)=\frac{2}{3}d_x(g)\,.
\end{equation}
Thus, to find $r_{x}(g)$, one can alternatively compute $d_x(g)$.

We next calculate $d_x(g)$ for any node in $\mathcal{K}_g$. For the convenience of description, we distinguish all nodes in $\mathcal{K}_g$, by giving a label sequence to every node, such that the nodes with identical label have identical properties, e.g., degree, resistance distance, and so forth. For this purpose, we
classify the $N_g$ nodes in $\mathcal{K}_g$ into $g+1$ different levels, with the $L_v(k)$ nodes created at the $k$th  iteration belonging to level $k$. By construction, for an arbitrary node $x$ at level $k$, there is a unique shortest path $v_0-v_1-v_2-\cdots-v_{n}-x$ (where $n\leq k-1$), from the nearest hub node $v_0$ at level $0$ to node $x$. We call $\{v_0, v_1, v_2, \cdots, v_{n}\}$ the set of ancestors for node $x$ and $v_{n}$ the parent of node $x$. Hence, each of the three hub nodes at level $0$ is the ancestors of $\frac{N_g}{3}-1$ nodes, and all other vertices are the descendant nodes of one of the three hub vertices at level $0$.

According to the unique shortest path to the closest hub node, we label each vertex at level $k$ by a unique sequence  $\{i_0, i_1, i_2, \cdots, i_{n}, k \}$, where $i_j$ is the level of node $v_j$ on the shortest path. It is easy to see that $i_0=0$ and  $0< i_1< i_2< \cdots< i_{n}< k$.  Although different nodes may have the same labels,  nodes with the same label have the same properties.
According to this labeling, for a node having label (sequence) $\{i_0, i_1, i_2, \cdots, i_{n} \}$, its parent has a label $\{i_0, i_1, i_2, \cdots, i_{n-1} \}$, and its ancestors have a label $\{i_0, i_1, i_2, \cdots, i_{k} \}$, where $k=0, 1, \cdots, n-1$.

As shown above, for any node $x$ in $\mathcal{K}_g$ it has a unique label $\{i_0,i_1\cdots,i_n\}$, where $i_0=0$ and $0< i_1< i_2<\cdots< i_{n}\leq g$, and for any pair of different vertices with identical labels, the value $d_x(g)$, or alternatively represented by $d_{\{i_0,i_1,\cdots,i_n\}}(g)$, for the two nodes is also identical. Based on the particular construction of $\mathcal{K}_g$, we can derive an explicit expression for $d_x(g)$ for any node $x$, as stated in the following lemma.
\begin{lemma}
For a node $x$ in the Koch networks $\mathcal{K}_g$ with label $\{i_0,i_1,\cdots,i_n\}$, its shortest-path distance  is
\begin{equation}\label{Si}
d_{\{i_0,i_1,\cdots,i_n\}}(g)=(g+2)4^{g}+2n\cdot4^g-\sum_{k=1}^{n}2\cdot4^{g-i_k}\,.
\end{equation}
\end{lemma}
\begin{proof}
By construction, for any node $x$ with label $\{0, i_1, i_2,\cdots, i_{n} \}$ at level $i_n$, there exists another vertex $x'$, called brother of $x$.  These two brothers are generated simultaneously, both of which and their  parent, denoted $p$, form a triangle. It is easy to see that the label of  $p$ is $\{0, i_1, i_2,\cdots, i_{n-1} \}$. To prove the lemma, we first derive the recursive relation  between $d_{\{0, i_1, i_2,\cdots, i_{n} \}}(g)$ for node $x$ and $d_{\{0, i_1, i_2,\cdots, i_{n-1} \}}(g)$ for its parent $p$.

Note that among the $N_g$ nodes in $\mathcal{K}_g$, $N_{g-i_n}/3$ nodes (including node $x$) are the descendants of $x$, which constitute a set $V^{x}_g$; $N_{g-i_n}/3$ nodes (containing node $x'$) are the descendants of node $x'$, which form a set $V^{x'}_g$; and the remaining $N_g-(2N_{g-i_n}/3)$ nodes form a set $V^{p}_g$. Then, $V_g=V^{x}_g \bigcup V^{x'}_g \bigcup V^{p}_g$ and $V^{x}_g \bigcap V^{x'}_g \bigcap V^{p}_g=\emptyset$. By definition, for a node $y \in V^{x}_g$, $d_{xy}(g)=d_{py}(g)-1$; for a node $j\in V^{x'}_g$, $d_{xy}(g)=d_{py}(g)$; and for a node $y \in V^{p}_g$, $d_{xy}(g)=d_{py}(g)+1$. Thus
\begin{align}
&d_x(g)=d_{\{0, i_1, i_2,\cdots, i_{n} \}}(g)= \sum_{y \in V_g}{d_{xy}(g)}  \nonumber \\
=&\sum_{y \in V^{x}_g}{d_{xy}(g)}+\sum_{y \in V^{x'}_g}{d_{xy}(g)}+\sum_{y \in V^{p}_g}{d_{xy}(g)}  \nonumber \\
=&\sum_{y \in V^{x}_g}(d_{py}(g)-1)+\sum_{y \in V^{x'}_g}{d_{py}(g)}+\sum_{y \in V^{p}_g}(d_{py}(g)+1) \nonumber \\
=&\sum_{y \in V^{x}_g}d_{py}(g)-\frac{1}{3}N_{g-i_n}+\sum_{y \in V^{x'}_g}{d_{py}(g)}+\sum_{j \in V^{p}_g}d_{py}(g)\nonumber\\
&\quad +N_g-\frac{2}{3}N_{g-i_n}  \nonumber \\
=&\sum_{y \in V_g}{d_{py}(g)} +N_g-N_{g-i_n}  \nonumber \\
=&d_p(g)+N_g-N_{g-i_n}  \nonumber \\
=&d_{\{0, i_1, i_2,\cdots, i_{n-1} \}}(g)+2\cdot4^g-2\cdot4^{g-i_n}\,.
\label{RSX}
\end{align}
Repeatedly applying~(\ref{RSX}), we obtain
\begin{align}
 &d_{\{0, i_1, i_2,\cdots, i_{n} \}}(g)\nonumber \\
=&d_{\{0, i_1, i_2,\cdots, i_{n-1} \}}(g)+2\cdot4^g-2\cdot4^{g-i_n}  \nonumber \\
=&d_{\{0, i_1, i_2,\cdots, i_{n-2} \}}(g)+4\cdot4^g-2\cdot4^{g-i_{n-1}}-2\cdot4^{g-i_n}  \nonumber \\
=&\cdots \nonumber \\
=&d_{\{0\}}(g)+2n\cdot4^g-\sum_{k=1}^{n}2\cdot4^{g-i_k}\nonumber \\
=&(g+2)4^{g}+2n\cdot4^g-\sum_{k=1}^{n}2\cdot4^{g-i_k}\,,
\label{ESX}
\end{align}
where $d_{\{0\}}(g)=(g+2)4^{g}$ derived in~\cite{ZhZhXiChLiGu09} was used.
\end{proof}

We note that Lemma~\ref{Si} was previously given in~\cite{YiZhShCh17}, where the proof is omitted.  Then, the resistance distance of a  node in  $\mathcal{K}_g$ can be directly determined by plugging~(\ref{Si}) to~(\ref{Hlfxg}).
\begin{lemma}
\label{lemmaA}
In the Koch networks $\mathcal{K}_g$, the resistance distance of node $x$ labelled by $\{i_0,i_1,\cdots,i_n\}$ is
\begin{equation}
\label{sumR}
r_{x}(g)=\frac{2 \times [(g+2)4^{g}+2n\cdot4^g-\sum_{k=1}^{n}2\cdot4^{g-i_k}]}{3}\,.
\end{equation}
\end{lemma}

Inserting (\ref{eqa:Kichoff}) and (\ref{sumR}) into~\eqref{EE04}, we obtain an exact expression for  $\LL^{\dagger}_{xx}(g)$ associated with any node in  $\mathcal{K}_g$.
\begin{theorem}
For the Koch network $\mathcal{K}_g$, the diagonal entry  $\LL^{\dagger}_{xx}(g)$ corresponding to node $x$ with label $\{i_0,i_1,\cdots,i_n\}$ is
\begin{align}\label{eqa:diagnol}
\LL^{\dagger}_{xx}(g)&=\frac{2}{3(2\cdot4^{ g}+1)^2}\left(2n\cdot4^g-\displaystyle \sum_{k=1}^{n}2\cdot4^{g-i_k}\right) \nonumber \\
&\quad +\frac{ 2^{2g+1}(5 \cdot 4^g+3g+4)}{9 (2\cdot4^{ g}+1)^2}.
\end{align}
\end{theorem}

\subsubsection{Ordering of nodes by   diagonal of pseudoinverse for  laplacian}
\label{sec:Ordering1}

Besides the explicit expression for the  diagonal entries of pseudoinverse for  Laplacian matrix, we can also compare the  diagonal entries for any pair of nodes in  $\mathcal{K}_g$.
\begin{theorem}
\label{the:compare}
For two different nodes $x_i$ and $x_j$ in $\mathcal{K}_g$, with their labels being, respectively,  $\{i_0,i_1,\cdots,i_{n_i}\}$ and $\{j_0,j_1,\cdots,j_{n_j}\}$ where $i_0=j_0=0$,
\begin{enumerate}
\item if $n_i>n_j$, then $\LL^{\dagger}_{x_ix_i}(g)>\LL^{\dagger}_{x_jx_j}(g)$;
\item if $n_i<n_j$, then $\LL^{\dagger}_{x_ix_i}(g)<\LL^{\dagger}_{x_jx_j}(g)$;
\item if $n_i=n_j$, and
\begin{enumerate}
\item if there exists a natural integer $z$ with $0\leqslant z\leqslant n_i$, such that $i_l=j_l$ for $l=0,1,2,\cdots,z-1$, but $i_z\neq j_z$,
\begin{enumerate}
\item if $i_z>j_z$, then $\LL^{\dagger}_{x_ix_i}(g)>\LL^{\dagger}_{x_jx_j}(g)$;
\item if $i_z<j_z$, then $\LL^{\dagger}_{x_ix_i}(g)<\LL^{\dagger}_{x_jx_j}(g)$;
\end{enumerate}
\item if $i_l=j_l$ for $l=0,1,2,\cdots,n_i$, then $\LL^{\dagger}_{x_ix_i}(g)=\LL^{\dagger}_{x_jx_j}(g)$.
\end{enumerate}
\end{enumerate}
\end{theorem}
\begin{proof}
Equality \eqref{EE04} implies that  the ordering of nodes by $\LL^{\dagger}_{xx}$ is consistent with the ordering of nodes by resistance distance $r_{x}(g)$. On the other hand,  the ordering of nodes by resistance distance $r_{x}(g)$ agrees with the the ordering of nodes by shortest-path distance $d_{x}(g)$. For any pair of nodes $x_i$ and $x_j$, in order to compare $\LL^{\dagger}_{x_ix_i}(g)$ and $\LL^{\dagger}_{x_jx_j}(g)$, we can alternatively compare $d_{x_i}(g)$ and $d_{x_j}(g)$.

According to~(\ref{Si}), we have
\begin{equation}
\begin{split}
& \quad d_{x_i}(g)-d_{x_j}(g)\\
&=d_{\{0, i_1, i_2,\cdots, i_{n_i} \}}(g)-d_{\{0, j_1, j_2,\cdots, j_{n_j} \}}(g)\\
&=2\cdot 4^g(n_i-n_j)-\sum_{k=1}^{n_i}2\cdot4^{g-i_k}+\sum_{k=1}^{n_j}2\cdot4^{g-j_k}.
\end{split}
\end{equation}
We can compare $d_{x_i}(g)$ and $d_{x_j}(g)$ by distinguishing three cases.
\begin{enumerate}
\item Case I: $n_i>n_j$. In this case, it is obvious that $n_i-n_j\geqslant 1$. Moreover, $g\geqslant n_i> n_j\geqslant 1$. Then, we get
\begin{equation}
\begin{split}
& \quad d_{x_i}(g)-d_{x_j}(g)\\
&=2\cdot 4^g(n_i-n_j)-\sum_{k=1}^{n_i}2\cdot4^{g-i_k}+\sum_{k=1}^{n_j}2\cdot4^{g-j_k}\\
&>2\cdot 4^g-2\,\sum_{k=0}^{g-1} 4^{k}=2\cdot 4^g-\frac{2}{3}(4^g-1)>0\,,
\end{split}
\end{equation}
which means $d_{x_i}(g)>d_{x_j}(g)$ and thus  $\LL^{\dagger}_{x_ix_i}(g)>\LL^{\dagger}_{x_jx_j}(g)$.
\item Case II: $n_i<n_j$. For this case, we can prove  $\LL^{\dagger}_{x_ix_i}(g)<\LL^{\dagger}_{x_jx_j}(g)$, by using a process similar to the first case $n_i>n_j$.
\item Case III: $n_i=n_j$.
For this particular case, we distinguish two subcases.
\begin{enumerate}
\item  If there exists a positive integer $z$ ($0< z\leqslant n_i$), such that $i_l=j_l$ for $l=0,1,2,\cdots,z-1$, but $i_{z}\neq j_{z}$, then
\begin{equation}
\begin{split}
& \quad d_{x_i}(g)-d_{x_j}(g)\\
&=2\cdot 4^g(n_i-n_j)-\sum_{k=1}^{n_i}2\cdot4^{g-i_k}+\sum_{k=1}^{n_j}2\cdot4^{g-j_k}\\
&=2\left(-\sum_{k=z}^{n_i}4^{g-i_k}+\sum_{k=z}^{n_j}4^{g-j_k}\right).
\end{split}
\end{equation}
\begin{enumerate}
\item If $i_{z} >j_{z}$, then
\begin{equation}
\begin{split}
& \quad d_{x_i}(g)-d_{x_j}(g)\\
&\geqslant 6\cdot4^{g-i_{z}}-2\sum_{k=0}^{g-i_{z}-1} 4^k >\frac{16}{3}\cdot 4^{g-i_{z}}>0\,,
\end{split}
\end{equation}
which implies  $\LL^{\dagger}_{x_ix_i}(g)>\LL^{\dagger}_{x_jx_j}(g)$.
\item If $i_{z} <j_{z}$, we can analogously prove that  $\LL^{\dagger}_{x_ix_i}(g)<\LL^{\dagger}_{x_jx_j}(g)$.
\end{enumerate}
\item If $i_l=j_l$ holds for $l=0,1,2,\cdots,n_i$, it is easy to prove that  $\LL^{\dagger}_{x_ix_i}(g)=\LL^{\dagger}_{x_jx_j}(g)$.
\end{enumerate}
\end{enumerate}
This completes the proof.
\end{proof}

Since the diagonal entry of $\LL^{\dagger}(g)$  is consistent with the structural centrality, current flow centrality, and information centrality of nodes in Koch networks $\mathcal{K}_g$, Theorem \ref{the:compare} shows that ranking nodes in $\mathcal{K}_g$
according to the values of diagonal for $\LL^{\dagger}(g)$  captures the information of all the three centralities for all nodes in $\mathcal{K}_g$.

\subsection{Diagonal  of  Pseudoinverse of Laplacian for Uniform Recursive Trees  }

This subsection is devoted to determining the diagonal  entries for the  pseudoinverse of Laplacian for uniform recursive trees~\cite{LiZh13}, which are also constructed in an iterative way. Let $\mathcal{U}_{g}$ ($g \geq 0$) be the networks after $g$ iterations. For $g=0$, $\mathcal{U}_{0}$ contains  an isolated node, called the central node. For $g=1$, $f$ ($f$ is a positive natural number) new nodes are created and linked to the central node to form $\mathcal{U}_{1}$. For $g \ge 1$, $\mathcal{U}_g$ is obtained from $\mathcal{U}_{g - 1}$  by performing the following operations: For each node in $U_{g-1}$,  $f$ new nodes are generated and attached to it.  Figure~\ref{graph1} illustrates the first several iterative constructions of a particular network for  $f=3$.

\begin{figure}[h]
\centering
\includegraphics[width=0.95\linewidth,trim=0 0 0 0]{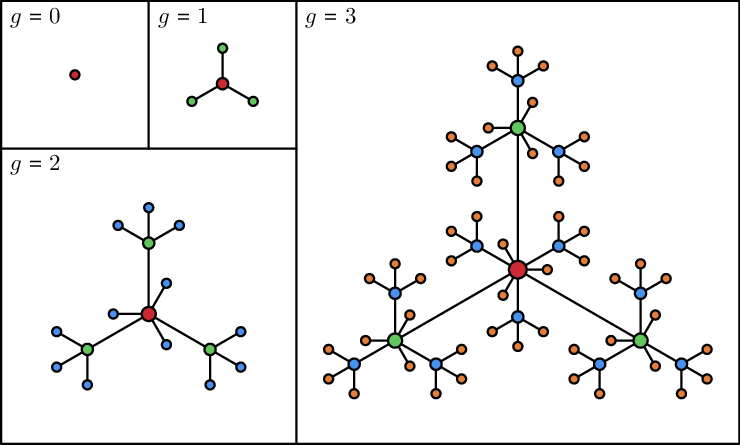}
\caption{ Construction process of a special network for $f=3$.} \label{graph1}
\end{figure}


In the case without inducing confusion, we use the same notations as those for Koch networks.
By construction,  at each iterative step $g_i$ ($g_i \ge 1$), the number of newly generated nodes is $L({g_i}) = f{(f
+ 1)^{{g_i} - 1}}$. Then in network $\mathcal{U}_{g}$, there are $N_g =(f + 1)^g$  nodes and ${E_g} =  {(f + 1)^g} - 1$ edges.
In contrast to Koch networks, the uniform recursive trees are not scale-free, but have an exponential form degree distribution. It has been observed~\cite{AmScBaSt00} that the degree distribution of some real-world networks also decays exponentially, such as power grid. In addition, the uniform recursive trees are small-world with a low average shortest distance.

The special construction of the uniform recursive trees allows to determine exactly relevant quantities.   It was shown~\cite{LiZh13} that the Kirchhoff index of  $\mathcal{U}_{g}$ is
\begin{equation}
\label{eqa:Kichoff2}
\delta(g)=(fg-1)(f+1)^{2g-1}+(f+1)^{g-1}\,.
\end{equation}

We next determine the diagonal  $\LL_{xx}^{\dagger}(g)$ of the pseudoinverse $\LL_{xx}(g)$ for the Laplacian matrix for graph   $\mathcal{U}_{g}$, based on the relation in~\eqref{EE04}. To do so, we need first find  $r_x(g)$
Note that for $\mathcal{U}_{g}$, $r_{ij}(g)=d_{ij}(g)$ for any pair of nodes, due to the tree-like structure.

In order to determine $r_x(g)$, similar to that of Koch networks, we classify the $N_g$ nodes in $\mathcal{U}_{g}$ into $g+1$ different levels, based on which  we provide a label to each node in $\mathcal{U}_{g}$.  For any node $x$ in $\mathcal{U}_{g}$,  we  label it by a sequence of node level information $\{0, i_1, i_2, ..., i_{n} \}$, $0< i_1< i_2< ...< i_{n}\leq g$. For any node $x$ with labeling $\{0, i_1, i_2, ..., i_{n} \}$, we interchangeably use  $x$ or $\{0, i_1, i_2, ..., i_{n} \}$ to represent this node.  For example, for the central node  at level $0$ in $\mathcal{U}_{g}$, its resistance distance  can be written as  $r_{\{0\}}(g)$.

Before determining  $r_x(g)$ for an arbitrary node $x$, we first  we calculate $r_{\{0\}}(g)$. For $g=1$, $r_{\{0\}}(1)=f$. For $g>1$, $r_{\{0\}}(g)$ satisfies the following recursion relation.
\begin{equation}
r_{\{0\}}(g) =r_{\{0\}}(g-1) +f \left(r_{\{0\}}(g-1)+N_{g-1}\right),  \nonumber
\end{equation}
which under the initial condition $r_{\{0\}}(1)=f$ is solved to yield
\begin{equation}\label{r0g}
r_{\{0\}}(g) =g \, f(f+1)^{g-1}.
\end{equation}

We are now in position to   calculate $r_{x}(g)$ for any node $x$ in $\mathcal{U}_{g}$. Note that for any node $x$ with label $\{0, i_1, i_2, ..., i_{n} \}$, $ 0< i_1< i_2< ...< i_{n}\leq g$, its parent, denoted by $p$, has a label  $\{0, i_1, i_2, ..., i_{n-1} \}$. We can derive the  relation between $r_{\{0, i_1, i_2, ..., i_{n} \}}(g)$ and $r_{\{0, i_1, i_2, ..., i_{n-1} \}}(g)$.
Notice  that  the descendants of  node $x$ at level $i_n$ and $x$ itself constitute one subunit of the whole network $\mathcal{U}_{g}$, which is a copy of $\mathcal{U}_{g-i_n}$. The node number  of this subunit is $N_{g-i_n}$. Let $\Omega$ denote the set of nodes in $\mathcal{U}_{g}$,  and let $\Omega_{x}$ denote the set of the descendant nodes of  $x$. Then,  $\Omega=\Omega_{x}\bigcup \overline{\Omega}_{x}$. By construction, for any node $y\in \Omega_{x}$, $r_{xy} (g)=r_{py} (g)-1$. While for any node $y\in \overline{\Omega}_{x}$, $r_{xy} (g)=r_{py} (g)+1$. Thus
\begin{align}
& \quad  r_{\{0, i_1, i_2, ..., i_{n} \}} (g)=r_x (g)= \sum_{y \in \Omega}{r_{xy} (g)}  \nonumber \\
&=\sum_{y \in \Omega_{x}}{r_{xy} (g)}+\sum_{y \in \overline{\Omega}_{x}}{r_{xy} (g)} \nonumber \\
&=\sum_{y \in \Omega_{x}}{(r_{py} (g)-1)}+\sum_{y \in \overline{\Omega}_{x}}{(r_{py} (g)+1)} \nonumber \\
&=\sum_{y \in \Omega_{x}}{r_{py} (g)}-N_{g-i_n}+\sum_{y \in \overline{\Omega}_{x}}{r_{py} (g)}+N_g-N_{g-i_n}  \nonumber \\
&=\sum_{y \in \Omega}{r_{py} (g)}+N_g-2N_{g-i_n}  \nonumber \\
&=r_p (g)+N_g-2N_{g-i_n}  \nonumber \\
&=r_{\{0, i_1, i_2, ..., i_{n-1} \}} (g)+(f+1)^g-2(f+1)^{g-i_n}. \nonumber \\
\label{RSX2}
\end{align}
Using the above relation \eqref{RSX2} repeatedly to give
\begin{align}
& \quad r_{\{0, i_1, i_2, ..., i_{n} \}}(g) \nonumber \\
&=r_{\{0, i_1, i_2, ..., i_{n-1} \}} (g)+\left[(f+1)^g-2(f+1)^{g-i_n}\right]  \nonumber \\
&=r_{\{0, i_1, i_2, ..., i_{n-2} \}}(g)+\left[(f+1)^g-2(f+1)^{g-i_{n-1}}\right]\nonumber\\
&\quad+\left[(f+1)^g-2(f+1)^{g-i_n}\right]  \nonumber \\
&=\cdots \cdots \nonumber \\
&=r_{\{0\}}(g)+\sum_{k=1}^{n} \left[(f+1)^g-2(f+1)^{g-i_k}\right],
\label{ESX2x}
\end{align}
which, together with $r_{\{0\}}(g) =g \, f(f+1)^{g-1}$, is  solved to obtain the following result.
\begin{lemma}
\label{lemmaB}
For a node $x$ in the uniform recursive tree $\mathcal{U}_g$ with label $\{i_0,i_1,\cdots,i_n\}$, its  resistance distance is
\begin{align}
\label{sumR2}
&\quad r_{x}(g)= r_{\{0, i_1, i_2, ..., i_{n} \}}(g) \nonumber \\
&= gf(f+1)^{g-1} + (f+1)^g [n - 2\sum_{k=1}^{n}(f+1)^{-i_k}]\,.
\end{align}
\end{lemma}

Plugging   (\ref{eqa:Kichoff2}) and (\ref{sumR2}) into~\eqref{EE04} yields an explicit expression for  $\LL^{\dagger}_{xx}(g)$ associated with any node in  $\mathcal{U}_g$.
\begin{theorem}
For  Laplacian matrix  $\LL(g)$ of the uniform recursive tree $\mathcal{U}_g$, the diagonal entry  $\LL^{\dagger}_{xx}(g)$ of its pseudoinverse  $\LL^{\dagger}(g)$ corresponding to node $x$ with label $\{i_0,i_1,\cdots,i_n\}$ is given by
\begin{equation}\label{eqa:diagnol2}
\LL^{\dagger}_{xx}(g)=n-2\sum_{k=1}^{n}(f+1)^{-i_k}+\frac{1}{f+1}-(f+1)^{-(g+1)}.
\end{equation}
\end{theorem}

In addition to the exact expression for the  diagonal of $\LL^{\dagger}(g)$  for $\mathcal{U}_g$, one can also compare the  diagonal entries between any two nodes in  $\mathcal{U}_g$ with  $f \geq 2$, as stated in the following theorem.
\begin{theorem}
\label{the:compare2}
In the case of  $f \geq 2$, for two different nodes $x_i$ and $x_j$ in the uniform recursive tree $\mathcal{U}_g$, with their labels being, respectively,  $\{i_0,i_1,\cdots,i_{n_i}\}$ and $\{j_0,j_1,\cdots,j_{n_j}\}$ where $i_0=j_0=0$,
\begin{enumerate}
\item if $n_i>n_j$, then $\LL^{\dagger}_{x_ix_i}(g)>\LL^{\dagger}_{x_jx_j}(g)$.
\item if $n_i<n_j$, then  $\LL^{\dagger}_{x_ix_i}(g)<\LL^{\dagger}_{x_jx_j}(g)$.
\item if $n_i=n_j$, and
\begin{enumerate}
\item if there exists a natural integer $z$ with $0\leqslant z\leqslant n_i$, such that $i_l=j_l$ for $l=0,1,2,\cdots,z-1$, but $i_z\neq j_z$,
\begin{enumerate}
\item if $i_z>j_z$, then  $\LL^{\dagger}_{x_ix_i}(g)>\LL^{\dagger}_{x_jx_j}(g)$;
\item if $i_z<j_z$, then  $\LL^{\dagger}_{x_ix_i}(g)<\LL^{\dagger}_{x_jx_j}(g)$;
\end{enumerate}
\item if $i_l=j_l$ for $l=0,1,2,\cdots,n_i$, then  $\LL^{\dagger}_{x_ix_i}(g)=\LL^{\dagger}_{x_jx_j}(g)$.
\end{enumerate}
\end{enumerate}
\end{theorem}
We omit the proof of Theorem~\ref{the:compare2}, since it is similar to that for Theorem~\ref{the:compare}.

\section{Numerical Results On Model Networks}

To further demonstrate the performance of our approximation algorithm $\text{Approx}\mathcal{D}$, we use it to compute the  diagonal elements of pseudoinverse of Laplacian matrix for some model networks, as well as their  Kirchhoff index.

\subsection{Diagonal  of Pseudoinverse for Laplacian}

We first apply our algorithm $\text{Approx}\mathcal{D}$ to approximately compute the diagonal elements of pseudoinverse of Laplacian matrix for the Koch network $\mathcal{K}_{10}$ and uniform recursive tree $\mathcal{U}_{11}$ with $f=3$.  Let $V$ be the set of nodes in these two graphs. Table~\ref{tab:DiagModel} reports the maximum relative error $\sigma_{\max}=\max_{u\in V}|\LL^{\dagger}_{uu}-\tilde{\LL}^{\dagger}_{uu}|/{\LL^{\dagger}_{uu}}$,  mean relative error $\sigma=\frac{1}{|V|}\sum_{u\in V}|\LL^{\dagger}_{uu}-\tilde{\LL}^{\dagger}_{uu}|/{\LL^{\dagger}_{uu}}$, and running time of our numerical results.

\begin{table}
	\tabcolsep=5pt
	\centering
	\fontsize{8.4}{8.4}\selectfont
	\begin{threeparttable}
\caption{Maximum relative error, mean relative error, and running time (seconds, $s$) for $\tilde{\calL^{\dagger}}$ on networks $\mathcal{K}_{10}$ and $\mathcal{U}_{11}$ with $f=3$. The diagonal elements of $\LL^{\dagger}(g)$ are obtained via~\eqref{eqa:diagnol} and~\eqref{eqa:diagnol2}, while $\tilde{\calL^{\dagger}}$ is obtained through algorithm $\mathtt{Approx\mathcal{D}}$ with $\epsilon=0.1$.}
		\label{tab:DiagModel}
		\begin{tabular}{cccccc}
			\toprule
			Network & Vertices & Edges & $\sigma_{\max}$ &  $\sigma$ & Time\cr
			\midrule
			\specialrule{0em}{3pt}{3pt}
			$K_{10}$    & 2,097,153 & 3,145,728 & 0.1366 & 0.0225 & 1630\cr
\specialrule{0em}{3pt}{3pt}
			$U_{11}$    & 4,194,304 & 4,194,303 & 0.1401 & 0.0220 &4496\cr
			\specialrule{0em}{3pt}{3pt}
			\bottomrule
		\end{tabular}
	\end{threeparttable}
\end{table}

\subsection{Kirchhoff Index}

We proceed to present our numerical results  for Kirchhoff index of three model networks, including Koch networks, uniform recursive trees,  and  pseudofractal scale-free webs~\cite{DoGoMe02,ShLiZh17,XiZhCo16}. For the first two graphs, the exact results for  Kirchhoff index have been obtained in the previous section, given by~\eqref{eqa:Kichoff} and~\eqref{eqa:Kichoff2}, respectively. Hence we only need to determine the Kirchhoff index for pseudofractal scale-free webs.

Let $\mathcal{F}_g$ ($g \geq 0$) denote the pseudofractal scale-free webs after $g$ iterations. For $g=0$, $ \mathcal{F}_0$ includes a triangle of three nodes and three edges. For $g>0$, $\mathcal{F}_g$ is obtained from $\mathcal{F}_{g-1}$ in the following way. For each existing edge in $\mathcal{F}_{g-1}$, a new node is created and connected to both end nodes of the edge.  Figure~\ref{psfw1} schematically illustrates the construction process of the graph. In  $\mathcal{F}_g$, there are $(3^{g+1}+3)/2$ nodes and $3^{g+1}$ edges. Like the Koch networks,  the pseudofractal scale-free webs also exhibit the striking scale-free small-world properties of many real networks.  It has been shown~\cite{YaKl15,YiZhSt19} that the Kirchhoff index $\delta(g) $ for $\mathcal{F}_g$ is
\begin{align}
\label{Kg01}
\delta(g)&=\frac{1}{112}\cdot\frac{1}{3^{g+2}}\times \nonumber\\
&\quad\big(50\cdot 3^{3g+3}-35\cdot 3^{2g+2}2^{g+1}+48\cdot 3^{2g+2}\nonumber\\
&\quad+30\cdot 3^{g+2}2^{g+1}-14\cdot 3^{g+2}+225\cdot 2^{g+1}\big)\,.
\end{align}

\begin{figure}
\begin{center}
\includegraphics[width=0.6\linewidth]{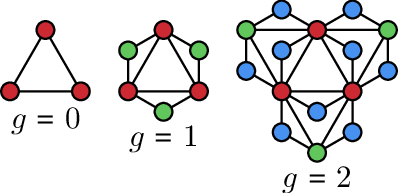}
\caption{Iterative construction process of the pseudofractal scale-free webs. }
\label{psfw1}
\end{center}
\end{figure}

Table~\ref{tab:Kirchoff} reports our numerical results for  Kirchhoff index on three model networks, $\mathcal{K}_{10}$, $\mathcal{U}_{11}$ with $f=3$,  and  $\mathcal{F}_{12}$.  All these  numerical results  are  obtained via algorithm $\text{Approx}\mathcal{D}$.

\begin{table}
	\tabcolsep=1pt
	\centering
	\fontsize{5.5}{4}\selectfont
	\begin{threeparttable}
		\caption{Exact Kirchhoff index $\delta$,   their approximation $\tilde{\delta}$,  relative error $\rho=(\delta-\tilde{\delta})/\delta$, and running time (seconds, $s$) for $\tilde{\delta}$ on three networks,  $\mathcal{K}_{10}$, $\mathcal{U}_{11}$ with $f=3$,  and  $\mathcal{F}_{12}$.  $\delta$ is obtained via~\eqref{eqa:Kichoff},~\eqref{eqa:Kichoff2}, and~\eqref{Kg01}, while $\tilde{\delta}$ is obtained through algorithm $\text{Approx}\mathcal{D}$ with $\epsilon=0.1$.}
		\label{tab:Kirchoff}
		\begin{tabular}{ccccccc}
			\toprule
			Net & Vertices & Edges & $\delta$ & $\tilde{\delta}$ &$\rho$ ($\times 10^{-3}$)& Time\cr
			\midrule
			\specialrule{0em}{3pt}{3pt}
			$K_{10}$            & 2,097,153 & 3,145,728 & 140,737,489,403,904 & 140,679,655,505,316 & 0.4109346 & 4496\cr
			\specialrule{0em}{3pt}{3pt}
			$U_{11}$    	    & 4,194,304 & 4,194,303 & 3,396,122,345,421 & 3,396,602,496,861 & 0.1413823 & 1155\cr
			\specialrule{0em}{3pt}{3pt}
		$\mathcal{F}_{13}$  & 2,097,153   & 3,145,728 & 16,370,506,924,032 & 16,346,238,029,818 & 1.4824766 & 1630\cr
			\specialrule{0em}{3pt}{3pt}
			\bottomrule
		\end{tabular}
	\end{threeparttable}
\end{table}

Tables~\ref{tab:DiagModel} and~\ref{tab:Kirchoff}  show that the approximation algorithm $\text{Approx}\mathcal{D}$ works effectively for all studied model networks. This again demonstrates the advantage of our proposed algorithm for large-scale networks.

\section{Conclusions}

The problem of estimating the diagonal entries of pseudoinverse $\LL^{\dagger}$ of the Laplacian matrix associated with a graph arises in various applications ranging from network science to social networks and control science. In this paper, we presented an inexpensive technique to approximately compute all the diagonal entries of $\LL^{\dagger}$ for a weighted undirected graph. Our method is based on the Johnson-Lindenstrauss Lemma and nearly-linear time solver, and has nearly linear time computation complexity. We analytically demonstrated that the approximation algorithm provides an approximation guarantee with error bounds. Moreover, we determined exactly all the diagonal entries of the pseudoinverse $\LL^{\dagger}(g)$ the Laplacian matrices for two iteratively growing networks, the scale-free small-world Koch networks and the uniform recursive trees. With the help of several recursive relations derived from the particular structure of the two networks, we obtained rigorous expressions for every entry of matrix $\LL^{\dagger}(g)$. We conducted extensive numerical experiments on various real-world networks and three model networks, the results of which demonstrate that the proposed approximation method is both effective and efficient. In future work, we plan to study nearly-linear time algorithm  estimating the diagonal entries of pseudoinverse  for a directed graph Laplacian~\cite{Bo21}, by using the approach of sparse LU factorizations developed in~\cite{CoKeDyPePeRaSi18}.

\bibliographystyle{compj}
\bibliography{Diagnol}

\end{document}